\documentclass{amsproc}
\usepackage{amsfonts}
\usepackage{amssymb}
\usepackage{amsmath}
\usepackage{lineno}
\newtheorem{theorem}{Theorem}[section]

\newcommand{\ff}{\mathbb{F}}

\newcommand{\de}{\delta}
\newcommand{\si}{\sigma}
\newtheorem{definition}[theorem]{Definition}
\newtheorem{example}[theorem]{Example}

\newtheorem{remark}[theorem]{Remark}
\newtheorem{proposition}[theorem]{Proposition}

\numberwithin{equation}{section}

\newcommand{\mulo}{K[t_1;\si_1,\de_1][t_2;\si_2,\de_2]\cdots[t_n;\si_n,\de_n]}

\begin{document}

\title{ Evaluation of iterated Ore polynomials and skew Reed-Muller codes}

\author{Andr\'e Leroy}
\address{ Université Artois, UR 2462, Laboratoire de Mathématiques de Lens (LML), F-62300 Lens, France}
\email{ andre.leroy@univ-artois.fr}

\author{ Nabil Bennenni}
\address{USTHB, Facult\'{e} de Math\'{e}matiques\\
	BP 32 El Alia Bab Ezzouar Algiers Algeria}
\email{nabil.bennenni@gmail.com}




\begin{abstract}
 In this paper, we study two ways of evaluating iterated Ore polynomials.  We provide many examples and compare these evaluations.  We use the evaluation maps to construct Reed-Muller codes and compute explicitly some of the data that are associated to such codes.
\end{abstract}

\maketitle

\section*{ Introduction}
The evaluation of polynomials is a keystone in many branches of mathematics.   In noncommutative algebra one of the most attractive notion of polynomials is the Ore polynomials (a.k.a. skew polynomials).  There are plenty of works dealing with the evaluations of these polynomials with coefficients in division rings or prime rings (\cite{L2},\cite{LLO}).  In coding theory the evaluation of polynomials is used, in particular, for constructing Reed-Muller codes (\cite{bib4}, \cite{bib1}, \cite{bib5}).

Let us recall the general definition of the skew polynomial ring.  
Let $K$ be a ring, $\si \in End(K)$, and $\de$ a skew $\si$-derivation on $K$ (i.e. $\de$ is an additive map from $K$ to $K$ such that for any $a,b\in K$, $\de(ab)=\si(a)\de(b)+\de(a)b$).   The elements of a skew polynomial ring $R=K[t;\si,\de]$ are polynomials $\sum_{i=0}^na_it^i$ with coefficients $a_i\in K$ written on the left.  
The addition of these polynomials is the usual one and the multiplication is based on the commutation rule $ta=\si(a)t + \de(a)$, for any $a\in K$.   With these two operations, $R$ becomes a ring.   When $K$ is a division ring, $R$ is a left principal ideal domain. If $\si$ is an automorphism of the division ring $K$ then $R$ is also a right principal ideal domain.  This construction can be iterated: we consider a ring $K$, a sequence of skew polynomial rings
$R_1=K[t_1;\si_1,\de_1]$, where $\si_1\in End(K)$ and $\de_1$ is a $\si_1$-derivation of $K$.  We then build
$R_2=R_1[t_1;\si_2,\de_2]=K[t_1;\si_1,\de_1][t_2;\si_2,\de_2]$,  where $\si_2\in End(R_1)$ and $\de_2$ is a $\si_2$-derivation of $R_1$.
Continuing this construction we arrived at
the iterated polynomial ring in $n$ noncommutative variables
$R=R_n=K[t_1;\si_1,\de_1][t_2;\si_2,\de_2]\cdots[t_n;\si_n,\de_n]$.  
The evaluation of skew polynomials has a very different behavior than the usual one for commutative variables.   Most of the literature is related to right evaluation of univariate (i.e. $n=1$) polynomials with coefficients in a division ring (e.g. \cite{ LLO}, \cite{L2}).  For the case when the base ring $K$ is not necessarily a division ring, the pseudo-linear transformations are very useful (\cite{L2}).   When $K$ is a division ring the left evaluation was also considered even when $\si$ is not onto \cite{lam}.  
While considering the evaluation of a polynomial $f(t_1,\dots,t_n)\in R=\mulo$ at a point $(a_1,\dots,a_n)\in K^n$, we are quickly facing a \lq \lq noncommutative phenomenon":  the ideal 
$I=R(t_1-a_1)+\cdots +R(t_2-a_2)$ can be equal to $R$.  We will show how it is possible to avoid this problem. \\
Reed-Muller codes over $\mathbb{F}_2$ were introduced by Reed in 1954 \cite{bib4}.   Kasami, Lin and Peterson in  \cite{bib2} extended the definition to any finite field $\mathbb{F}_q$.  Reed-Muller codes are very useful in cryptography. Indeed, they show good resistance to linear cryptanalysis.   The interested reader may consult  \cite{bib5}, for more details. Reed-Muller codes are amongst the simplest examples of the class of geometrical codes,  which also includes Euclidean  geometry and projective geometry codes.  

The paper is organized as follows. Section 2 presents two ways of evaluating iterated polynomials in a general context.  We provide many examples and compare these evaluations.  
In Section 3, we apply the evaluation maps to construct 
Reed-Muller codes and compute some of the data that are associated to such codes.

%
%
\section{Polynomial maps in $n$ variables.} 

In this section, we define a notion
of evaluation for polynomials in an iterated Ore extension $R=\mulo$ over an arbitrary ring $K$ (with unity).   Throughout the paper we require that $\si_i(K)\subset K$ and $\de_i(K)\subset K$ for any $1\le i \le n$.
We will need some facts about iterated polynomial rings.
\begin{remark}
	\label{First remarks} 
	\label{inner auto and inner derivations}
	{\rm 
		\begin{enumerate}
			\item When $\de=0$ we denote the Ore extension as $R=K[t;\si$].  Similarly when $\si=id$, we write $R=K[t;\de]$.
			\item We remark that if $u\in K$  an invertible element and, for any $x\in K$, we have $\si(x)=I_u(x)=uxu^{-1}$ then $K[t;\si;\de]
			=K[u^{-1}t;u^{-1}\delta]$.  Notice that if $\delta=0$, this gives a usual polynomial ring.
			\item If $\si$ is an endomorphism of the ring $K$ and $x,a$ are elements in $K$, we define $\de_a(x)=ax-\si(x)a$.   This is a $\si$-derivation on $K$ called the inner derivation on $K$.
			Let us notice that in this case we have, for any $x\in K$ that
			$(t-a)x=\si(x)(t-a)$ so that $R=K[t;\si,\de]=K[t-a;\si]$.
			\item Combining the two previous remarks we have that, for any $u,v\in K$ with $u$ an invertible element,  $K[t;I_u,\de_v]=K[u^{-1}(t-v)]$ is a usual polynomial ring.
		\end{enumerate}
	}
\end{remark}

Let us briefly recall the definition of the evaluation of a polynomial $f(t)\in R=K[t;\si,\de]$ at an element $a \in K$.
We define $f(a)\in K$ to be the only element of $K$ such that $f(t)-f(a)\in R(t-a)$.   This means, in particular, that $a\in K$ is a right root of $f(t)$ when $t-a$ is right factor of $f(t)$.   We then introduce, for any $i\ge 0$, a  map $N_i$ defined by induction as follows:
$$
N_0(a)=1,\quad N_{i+1}(a)=\si(N_i(a))a + \de(N_i(a)). 
$$ 
This leads to a concrete formula for the evaluation of any polynomial $f(t)=\sum_{i=0}^nb_it^i\in R=K[t;\si,\de]$ at an element $a\in K$:
$$
f(a)=\sum_{i=0}^{n}b_iN_i(a).
$$

Before defining the evaluation of iterated polynomials, we now provide a few classical examples.
\begin{example}
	{\rm 
		\begin{enumerate}
		\item[(1)] Consider the complex number $\mathbb{C}$ and the conjugation denoted by \lq\lq $-$".  We then construct the Ore extension: $R=\mathbb{C} [t;-]$; the commutation rule is here $t(a+ib)=(a-ib)t$, where $a,b\in \mathbb R$.  An element $a\in \mathbb C$ is a (right) root of $t^2+1$
		if $N_2(a) +1=0$ i.e. $\overline{a}a+1=0$.  From this it is clear that $t^2+1$ is an irreducible polynomial in $R$.   On the other hand, it is easy to check that $ t^2+1$ is a central polynomial and we get that the quotient $R/(t^2+1)$ is isomorphic to the quaternion algebra $\mathbb H$.  On the other hand, the roots of the polynomial $t^2-1$ are exactly the complex numbers of norm $1$.
		\item[(2)] Another important kind of Ore extension is obtained by presenting the Weyl algebra $A_1=k[X][Y;id.,\frac{d}{dx}]$, where $k$ is a field.  The commutation rule that comes up is thus $YX=XY +1$.  In characteristic zero this algebra is simple (it doesn't have any two sided ideal except $0$ and $A_1$).
		If $char(k)=p>0$, $A_1$ is a $p^2$ dimensional algebra
		over its center $k[X^p,Y^p]$.  Evaluating powers of $Y$ at $X$ we get for instance $Y^2(X)=X^2+1$, $Y^3(X)=X^3+3X$, and $Y^4(X)=X^4+6X^2+3$.
		
		\noindent We can iterate the procedure and get the Weyl algebra $A_n(k)$.
		
		\item[(3)] In coding theory, Ore extensions of the form  $\mathbb F_q[t;\theta]$, where $q=p^n$ and $\theta$ is the Frobenius automorphism defined by  $\theta(a)=a^p$, for $a\in \mathbb F_q$, have been extensively used (cf. \cite{BU}).  In \cite{L2}, both the evaluation and the factorization for these extensions were described in term of classical (untwisted) evaluation of polynomials.  General Ore extensions have also been used for constructing codes (\cite{BL}).
	\end{enumerate}
	}
\end{example}

Let $f(t_1,t_2,\dots,t_n)\in R=\mulo$ and an $n$-tuple $(a_1,a_2,\dots , a_n)\in K^n $ be given ($K$ is a ring which is assumed to be stable by the maps $\si_i$ and $\de_i$, with$1\le i \le n$).  We want to define the value of $f(a_1,a_2,\cdots,a_n)$.
It will be convenient to have notations for intermediate Ore extensions.  So we put  $R_0=K$, and, for $i\in \{1,\dots,n-1\}$,
$R_{i}=R_{i-1}[t_{i};\sigma_i,\delta_i]$.   Notice that $R_n=R$.  

Let $(a_1,a_2,\cdots,a_n)\in K^n$ and consider
$f_0=f(t_1,\dots,t_n)\in R=R_{n-1}[t_n;\sigma_n,\delta_n]$.  We can define $f_1=f(t_1,\dots,t_{n-1},a_n)\in R_{n-1}$ as the remainder of the division of $f_0$ on the right by $t_n-a_n$.  This is 
$$
f_0= q(t_1,\dots,t_{n-1})(t_n-a_n)+f_1.
$$
By the same procedure we divide $f_1=f(t_1,\dots,t_{n-1},a_n)$ by $t_{n-1}-a_{n-1}$ and get the remainder $f_2=f(t_1,\dots, t_{n-2},a_{n-1},a_n)\in R_{n-2}$.   This is
$$
f_1=f(t_1,\dots,t_{n-1},a_n)=q_2(t_1,t_2,\dots, t_{n-1})(t_{n-1}-a_{n-1})+f_2.
$$
Continuing this process, we define $f_{3}\in R_{n-3},\dots f_{n-1}\in R_1, f_n\in R_0=K$.  The evaluation of $f(t)$ at $(a_1,\dots,a_n)$ is $f_n\in K$.




This leads to the following definition.
\begin{definition}
	For $f(t_1,\dots, t_n) \in R= \mulo$ and $(a_1,\dots,a_n)\in K^n$, we define $f(a_1,\dots,a_n)$, 
	the evaluation of 
	$f(t_1,t_2,\dots, t_n)$ at $(a_1,\dots,a_n)$, as the representative in $K$ of $f(t_1,t_2,\dots, t_n)$ modulo  
	$I_n(a_1,\dots,a_n)=R_1(t_1-a_1)+ \dots +R_{n-1}(t_{n-1}-a_{n-1})+R(t_n-a_n)$, where, for $1\le i \le n$, $R_i$ stands for $R_i=K[t_1;\si_1,\de_1]\cdots [t_i;\si_i,\de_i]$. 
\end{definition}

Remark that $I_n=I_n(a_1,\dots,a_n)$ is an additive subgroup of $R,+$ and is not, in general, an ideal in $R=R_n$.  

Another way to compute the evaluation is to remark that the polynomials in $R=\mulo$ can be written in a unique way as sums of monomials of the form $\alpha_{l_1,\dots ,l_n}t_1^{l_1}t_2^{l_2}\cdots t_n^{l_n}$, for some $0\le l_1,l_2,\dots,l_n\le n$ and $\alpha_{l_1,\dots ,l_n}\in K$. The sum $\sum_{i=1}^nl_i$ is called the degree of the monomial and the degree of a polynomial is given by the degree of the monomials with higher degree.  
First remark that if a monomial $m=m(t_1,\dots , t_n)$ is of degree $l$, then for any $a\in K^{*}$, $ma$ is a polynomial of degree $l$ as well.  We define the evaluation by induction on the degree.
In practice, it is sufficient to define the evaluation of a monomial.

For $(a_1,a_2,\dots,a_n)\in K^n$ and $m=m(t_1,\dots , t_n)$ we define $m(a_1,\dots , a_n)$ as follows:

If $deg(m)=1$ and $m=\alpha t_i$ for some $1\le i \le n$ and $\alpha \in K$ then  $m(a_1,\dots , a_n)=a_i$.   
So, assume that the evaluation of the monomials of degree $<l$ has been defined and consider a monomial $m=m(t)$ such that
$deg(m)=l$ and $m=m'(t_1,\dots,t_j)t_j$ for some $1\le j \le n$,
then $m(a_1,\dots,a_n)=m''(a_1,\dots,a_j)$ where the polynomial 
$m''(t_1,\dots,t_j)=m'(t_1,\dots,t_j)a_j$ is of degree smaller than $l$.

Let us make some remarks about this evaluation.

\begin{remark}
	{\rm 
		\begin{enumerate}
			\item First let us remark that, before we evaluate a polynomial, we must write it with the variables appearing in the precise order $t_1,t_2,\dots,t_n$ (from left to right).  In other words, before evaluating a polynomial we must write it as a sum of monomials of the form $t_1^{l_1}t_2^{l_2}\cdots t_n^{l_n}$.  
			\item Of course, this is not the only possible definition, but although
			it might look a bit strange, still it is natural if we want the
			zeros being right roots.  Let us look more closely to the case of two variables.  In other words it is natural for $(a_1,a_2)$
			to annihilate a polynomial of the form $g(t_1)(t_2-a_2)$;  but we
			don't necessarily expect $(a_1,a_2)$ to be a zero of a polynomial of
			the form $(t_1-a_1)h(t_1)$.
			
			\item Since the base ring is not assumed to be commutative we must be very cautious while evaluating a polynomial, even when the variables commute.   
			With this definition, $t_1t_2\in K[t_1,t_2]$ evaluated
			at $(a,b)$ gives $(t_1t_2)(a,b)= ba$.  This might look very
			strange but if we think of evaluation in terms of \lq\lq operators" via
			the right multiplication by $b$ followed by the right
			multiplication by $a$ this evaluation looks perfectly fine and
			the apparent awkwardness disappears.  
			\item We assumme that the different endomorphisms $\si_i$s are such that $\si_i(K)\subseteq K$.  In other words, for $i>1$,
			$\si_i$ is an extension of $\si_i \vert _K$ to $R_{i-1}$.  Some commutation relations exist between the different endomorphisms $\si_i$.
			For instance, consider the polynomial ring extension
			$R=K[t_1;\si_1][t_2;\si_2]$.  If we put $\si_2(t_1)=\sum_{i=0}^{l}a_it^i$ computing $\si_2(\si_1(a)t_1)=\si_2(t_1a)=\si_2(t_1)\si_2(a)$, leads to the following equations
			$$
			\forall \;\; 0\le i \le l, \quad a_i\si_1^i(\si_2(a))=\si_2(\si_1(a))a_i.
			$$ 
			%
			%
			%
		\end{enumerate}	
	}
\end{remark}
\vspace{3mm}

\noindent Let us now give some examples.

\begin{example}
	
	{\rm
		
		\item[(1)] Let $A_1(k)=k[X][Y;id.,\frac{d}{dX}]$ and $(a,b)\in K^2$, then
		\begin{itemize}
			\item $YX=XY+1=X(Y-b)+bX+1=X(Y-b)+b(X-a)+ba+1$ and hence $(YX)(a,b)=ba + 1$.
			\item $YX^2=X^2Y+2X=X^2(Y-b)+bX^2+2X=X^2(Y-b)+bX(X-a)+baX+2(X-a)+2a$ and hence $(YX^2)(a,b)=ba^2+2a$.
			\item $Y^2X=XY^2+2Y=XY(Y-b)+bX(Y-b)+bXb+2(Y-b)+2b$ and hence $(Y^2X)(a,b)=b^2a + 2b$.
		\end{itemize}
		\item[(2)] Consider the double Ore extension $R=\mathbb F_q[t_1;\theta][t_2;\overline{\theta}]$, where $q=p^n$, $\overline{\theta}(a)=a^p$, and $\overline{\theta}(t_1)=t_1$.  A polynomial $p(t_1,t_2)\in R$ can be written as $p(t_1,t_2)=\sum_{i=0}^{n}p_i(t_1)t_2^i=\sum_{i,j}a_{i,j}t_1^jt_2^i$ and we can easily check that $$p(t_1,t_2)(a,b)=\sum_{i,j}\theta^j(N_i(b))N_j(a)=b^{\frac{(p^i-1)p^j}{p-1}}a^{\frac{p^i-1}{p-1}}.$$
	}	 
\end{example}

Let us now turn to another possible way of evaluating a polynomial $f(t_1,\dots, t_n)\in R=\mulo$ at $(a_1,a_2, \dots, a_n)\in K^n$.   We consider the element of $K$ 
representing $f$ in the quotient $R/I$ where 
$I=R(t-a_1)+R(t-a_2)+\dots + R(t-a_n)$.   The set $I$ is the usual left ideal of $R$ and this evaluation looks more 
classical.   We now compare the two evaluations by 
comparing $I_n=R_1(t_1-a_1)+ \dots +R_{n-1}(t_{n-1}-a_{n-1})+R(t_n-a_n)$ and $I$. 
%
%
%
%
%
\begin{theorem}
	Let $K$ be a ring and $R=\mulo$.  We consider $(a_1,a_2,\dots,a_n)\in K^n$ and put $I=R(t-a_1)+R(t-a_2)+\dots + R(t-a_n)$ and $I_n=R_1(t_1-a_1)+ \dots +R_{n-1}(t_{n-1}-a_{n-1})+R(t_n-a_n)$, where, for $1\le i \le n$, $R_i$, $R_i=K[t_1;\si_1,\de_1]\cdots [t_i;\si_i,\de_i]$.   With these notations, the following are equivalent:
	\begin{enumerate}
		\item $I_n=I$;
		\item $R(t_i-a_i)\subseteq I_n$;
		\item $I\ne R$;
		\item For $1\le i <j \le n$, we have $t_j(t_i-a_i)\in I_n$;
		\item For $1\le i <j \le n$, we have $\si_j(t_i-a_i)a_j + \de_j(t_i-a_i)\in I_n$;
		\item For $1\le i <j \le n$, we have $(t_jt_i)(a_1,\dots,a_n)=\si_j(a_i)a_j+\de_j(a_i)$.
		
	\end{enumerate}
\end{theorem}
\begin{proof}
	\noindent $1) \Leftrightarrow 2)$ is clear.
	
	\noindent $2) \Rightarrow 3)$ If $I=R$, then $I_n=R$ and, for $1\le i \le n$, there
	exist polynomials $g_i(t_1,t_2\cdots,t_i)\in R_i$ such that
	$1=\sum_{i=1}^ng_i(t_1,\dots ,t_i)(t_i-a_i)$.  The change of variables defined by putting $y_i=t_i-a_i$ gives that, for some $h_i\in R_i$, we have $1=\sum_{i=1}^nh_i(y_1,\dots ,y_i)y_i$.  A comparison of the coefficients of degree zero of this equality leads to a contradiction. 
	
	\noindent $3) \Rightarrow 4)$ There exists $c\in K$ such that
	$t_j(t_i-a_i)-c\in I_n\subseteq I$.  Since $t_j(t_i-a_i)\in I$ we
	obtain $c\in I$ hence, by $3)$, we must have $c=0$.  This shows
	that $t_j(t_i-a_i)\in I_n$.
	
	\noindent $4) \Rightarrow 5)$ The fact that $t_j(t_i-a_i)\in I_n$,
	implies that $\si_j(t_i-a_i)t_j+\de_j(t_i-a_i)\in I_n$.  Since
	$(t_j-a_j)\in I_n$, we obtain that
	$\si_j(t_i-a_i)a_j+\de_j(t_i-a_i)\in I_n$.
	
	\noindent $5) \Rightarrow 6)$ This is clear since $5)$ implies
	that
	$t_jt_i-\si_j(a_i)a_j-\de_j(a_i)=\si_j(t_i)t_j+\de_j(t_i)-\si_j(a_i)a_j-\de_j(a_i)\in
	I_n$.
	
	\noindent $6) \Rightarrow 1)$ The equality in $6)$ implies that
	$t_jt_i-\si_j(a_i)a_j-\de_j(a_i)\in I_n$ for $1\le i < j\le n$. This gives that $t_j(t_i-a_i)\in I_n$.  Since also we have that $R_i(t_i-a_i)\subseteq I_n$ we conclude that $t_j(t_i-a_i)\in I_n$
	for any integer $1\le i \le j \le n$ and hence $R(t_i-a_i)\subseteq I_n$ for any $1\le i \le n$.  This yields $I_n=I$, as required.
\end{proof}

\begin{remark}
	{\rm 
		1) If $n=1$, we obviously have $I=I_1$ and all points are good.
		
		2) In general, the two additive subset $I_n\subset I$ are different.  As mentioned above, we will use $I_n$ for our evaluation. The reason is that while evaluating with respect to $I$ we often face the following problem: the left $R$-module $I$ can be the entire ring.  So that the evaluation of any polynomial at $(a_1,a_2\dots,a_n)\in K^n$ with respect to $I$ is zero.  This the case in the Weyl algebra
		$R=A_1(K)=K[t_1][t_2;id,\frac{d}{dt_1}]$ for the piont $(0,0)$ since we then have $t_2t_1-t_1t_2=1$ and hence $Rt_1+Rt_2=R$.  This is not the case with our evaluation since, for instance, $t_2t_1=t_1t_2 +1 $,  so that $t_2t_1 +I_2(0,0) = 1 + I_2(0,0)$ and hence the evaluation of $t_2t_1$ at $(0,0)$ is just $1$.
		
		3) In fact,  it is quite often the case that $I=R$, even if we are using Ore polynomials with zero derivations. Consider for instance the Ore extension $R=K[t_1;\si_1][t_2;\si_2]$ where $K$ is a field and $\si_2$ is an endomorphism of $K[t_1;\si_1]$ such that $\si_2(t_1)=t_1$.  It is easy to check that for any $(a_1,a_2)\in K^2$ we have 
		$(t_2-\si_1(a_2))(t_1-a_1)+(-t_1+\si_2(a_1))(t_2-a_2)=\si_1(a_2)a_1-\si_2(a_1)a_2$.  So that if $\si_1(a_2)a_1-\si_2(a_1)a_2\ne 0$, then the left ideal $I(a_1,a_2)=R$.  This shows that very often the evaluation modulo $I$ turns out to be trivial.  Once again, this is not the case with our evaluation, since we have that 
		$t_2(t_1-a_1)$ is represented by  $\si_1(a_2)a_1-\si_2(a_1)a_2$ modulo $I_2(a_1,a_2)$.
	} 
\end{remark}

\begin{definition}
	A point $(a_1,\dots,a_n)\in K^n$ will be called a good point if the two ways of evaluating a polynomial $in \mulo$ at $(a_1,\dots,a_2)$ coincide i.e. if $I_n(a_1,\dots,a_n)= I$.
\end{definition}

The advantage of the good points is that in this case  the evaluations via a left ideal $I$ and $I_n$ coincide.
But of course, we can still evaluate a polynomial at any points via our additive subset $I_n=R_1(t_1-a_1)+ \dots +R_{n-1}(t_{n-1}-a_{n-1})+R(t_n-a_n)$.

\begin{example}
	\begin{enumerate}
		\item In the classical case ($\si_i=id_K$ and
		$\de_i=0$, for every $1\le i \le n$), every point $(a_1,\dots,a_n)\in K^n$ is good.
		\item If $K$ is a division ring $\si_1=id_K,\; \de_1=0$ and $\si_2=id,
		\de_2=d/dt_1$, we have for any $a,b\in K,\;
		(t_2-b)(t_1-a)=(t_1-a)(t_2-b)+1$.  This shows that in this case
		there are no good points.
	\end{enumerate}
\end{example}

Working with $I_n$ instead of $I$ we avoid the problem of having points that are zeroes of every polynomial in the ring $R=\mulo$.   

We will consider Reed-Muller codes and hence we will need to evaluate polynomials in several variables.  We consider the case of iterated Ore polynomials defined on a finite base ring (field) $K$.  In classical Reed-Muller coding the polynomials that are used to make evaluations are restricted to be of some bounded degree in each variables.   Indeed polynomial maps associated to $X^q-X\in \ff_q[X]$ are identically zero.  So our first task is to consider what is the analogue of this polynomial in an Ore polynomial ring.   We will need a few definitions.

\begin{definition}
	Let $a\in K$ be an element of a division ring $K$, $\si$ and $\de$ an endomorphism of $K$ and a 
	$\si$-derivation of $K$, respectively.
	\begin{enumerate}
		\item For a nonzero $x\in K$, we denote  $a^x=\si(x)ax^{-1} + \de(x)x^{-1}$ and put $\Delta(a)=\{a^x \mid x\in K^*\}$.  This set is called the ($\sigma,\delta$) conjugacy class of $a$.
		\item The map $T_a:K \longrightarrow K$ defined by  $T_a(x)=\si(x)a +\de(x)$ is called the ($\si,\de$) pseudo-linear map associated to $a$.
		\item The ($\si,\de$) centralizer of $a$ is the set $C^{\si,\de}(a)=\{x\in K^* \mid a^x=a\}\cup \{0\}$.
	\end{enumerate} 
\end{definition} 
Let us mention that the notion of $\si,\de$ conjugation appears naturally due to the fact that while evaluating a product $fg$ we have the nice formula $fg(a)=0$ if $g(a)=0$ and $(fg)(a)=f(a^{g(a)})g(a)$ if $g(a)\ne 0$.  Let us also mention that the evaluation at an element $a\in K$ is strongly related to the ($\si,\de$)-pseudo linear map $T_a$ via the equality $f(a)=f(T_a)(1)$ or more generally that, for any $h\in R=K[t;\si,\de]$, $a\in K$ and $x\in K^*$,  $h(T_a)(x)=h(a^x)x$. 
We first need the following lemma which is part of folklore. 

\begin{theorem}
	\label{generalization of Gordon-Motzkin}
	
	Let $a$ be an element of a division ring $K$ and $f(t)\in R=K[t;\sigma,\delta]$ be a polynomial of degree $n$.
	Then:
	\begin{enumerate}
		\item[(1)] $C^{\si,\de}(a)$ is a subdivision ring of $K$.
		\item[(2)] The map $T_a$ is a right $C^{\si,\de}(a)$ linear map.
		\item[(3)] $f(t)$ has roots in at most $n$ $(\sigma,
		\delta)$-conjugacy classes, say $\{\Delta
		(a_1),\dots,\Delta(a_r)\}$, $r\le n$;
		\item[(4)] $\sum_{i=1}^rdim_{C(a_i)}\ker (f(T_{a_i}))\le n$, where 
		$C(a_i):=C^{\si,\de}(a_i)$ for $1\le i \le r$.
	\end{enumerate}
\end{theorem}
\begin{proof}
	We refer the reader to \cite{LLO} and \cite{LO}.
\end{proof}

Notice that point (4) above generalizes several classical theorems, in particular, the theorem by Gordon-Motzkin that states that a polynomial $f(X)\in K[X]$ with coefficients in a division ring $K$ can have roots in at most
$n=deg (f)$ classical ($\si=id., \de=0$)-conjugacy classes.   
Theorem \ref{finite fields} below can be found in \cite{L2}.

\begin{theorem}
	\label{finite fields}	
	Let $p$ be a prime number and $\mathbb F_q$ be the finite field
	with $q=p^n$ elements.  Denote by $\theta$ the Frobenius
	automorphism.  Then:
	\begin{enumerate}
		\item[a)] There are $p$ distinct $\theta$-conjugacy classes in $\mathbb F_q$.
		\item[b)] $C^{\theta}(0)=\mathbb F_q$ and, for $0\ne a\in \mathbb F_q$, we have
		$C^{\theta}(a)=\mathbb F_p$.
		\item[c)] In $\mathbb F_q[t;\theta]$, the least left common multiple of all the elements
		of the form $t-a$ for $ a\in\mathbb F_q$ is the polynomial
		$G(t):=t^{(p-1)n +1}-t$.  In other words, $G(t)\in \mathbb
		F_q[t;\theta]$ is of minimal degree such that $G(a)=0$ for all
		$a\in \mathbb F_q$.
		\item[d)] The polynomial $G(t)$ obtained in c) above is invariant,
		i.e. $RG(t)=G(t)R$.
	\end{enumerate}
\end{theorem}
\begin{proof}
	a)   Let us denote by $g$ a generator of the cyclic group $\mathbb
	F_q^{*}:=\mathbb F_q \setminus\{0\}$. The $\theta$-conjugacy class
	determined by the zero element is reduced to $\{0\}$ i.e. $\Delta (0)=\{0\}$.
	The $\theta$-conjugacy class determined by $1$ is a subgroup of $\mathbb
	F_q^*$: $\Delta (1)=\{\theta (x)x^{-1}\,|\, 0\ne x\in \mathbb F_q
	\}=\{x^{p-1}\,|\, 0\ne x\in \mathbb F_q\}$.  It is easy to check
	that $\Delta (1)$ is cyclic generated by $g^{p-1}$ and has order
	$\frac{p^n-1}{p-1}$.  Its index is $(\mathbb F_q^* : \Delta
	(1))=p-1$.   Since two nonzero elements $a,b$ are
	$\theta$-conjugate if and only if $ab^{-1}\in \Delta (1)$,  we
	indeed get that the number of different nonzero $\theta$-conjugacy classes
	is $p-1$.   This yields the result.
	
	\noindent b) If $a\in \mathbb F_q$ is nonzero, then
	$C^{\theta}(a)=\{x\in \mathbb F_q\,|\, \theta(x)a=ax\}$ i.e.
	$C^{\theta}(a)=\mathbb F_p$.
	
	\noindent c)  We have, for any $x\in \mathbb F_q,\,
	(t^{(p-1)n+1}-t)(x)=\theta^{(p-1)n}(x)\dots \theta(x)x - x$. Since
	$\theta^n=id$, and $N_n(x):=\theta^{n-1}(x)\dots
	\theta(x)x \in \mathbb
	F_p$, we get $(t^{(p-1)n+1}-t)(x)=
	x(\theta^{n-1}(x)\dots\theta(x)x )^{p-1}-x=
	xN_n(x)^{p-1}-x=0$.  This shows that indeed $G(t)$ 
	annihilates all the elements
	of $\mathbb F_q$ and hence $G(t)$ is a  left common multiple of
	the linear polynomials $\{(t-a)\,|\,a\in \mathbb F_q\}$.  Let $h(t):=[t-a\,|\,a\in \mathbb F_q]_l$ denote their least left common multiple. It remains to show that
	$\deg h(t)\ge n(p-1)+1$.  Let $0=a_0,a_1,\dots,a_{p-1}\in \mathbb{F}_q$ be
	elements  representing the $\theta$-conjugacy classes (Cf. a) above).  Denote by $C_0,C_1,\dots,C_{p-1}$ their respective
	$\theta$-centralizer.  The formulas recalled in the paragraph before  Theorem \ref{generalization of Gordon-Motzkin} shows that 
	$h(T_a)(x)=h(a^x)x=0$ for any nonzero element $x\in \mathbb F_q$ and any element 
	$a\in\{a_0,\dots,a_{p-1}\}$.   
	Hence $\ker h(T_{a_i})=\mathbb F_q$ for $0\le i \le p-1$. Using
	Inequality $(4)$ in Theorem \ref{generalization of
		Gordon-Motzkin} and the statement $b)$ above, we get $\deg h(t)\ge
	\sum_{i=0}^{p-1} dim _{C_i}\ker h(T_{a_i})=dim_{\mathbb F_q}\mathbb
	F_q + \sum_{i=1}^{p-1}dim_{\mathbb F_p}\mathbb F_q =1+(p-1)n$, as
	required.
	
	\noindent d)  Since $\theta^n=id.$, we have immediately that
	$G(t)x=\theta(x)G(t)$ and obviously $G(t)t=tG(t)$.
\end{proof}

Let us now consider an iterated Ore extension $\mulo$, where $K$ is a finite ring.   Using classical notations from the world of algebraic geometry, we consider 
$$
I(K^n)=\{f(t_1,\dots,t_n)\in R \mid \forall (a_1,\dots,a_n)\in K^n, f(a_1,\dots,a_n)=0\}.
$$
In other words: $I(K^n)=\cap_{(a_1,\dots,a_n)}I_n(a_1,\dots,a_n)$.
In the case of an Ore extension 
$\mathbb F_q[t_1;\theta]\dots [t_n;\theta]$, where $\theta$ is the Frobenius automorphism of $\mathbb F_q=\mathbb F_{p^n}$, the above theorem shows that, for any $1\le i \le n$,  $t_i^{(n-1)p}-t_i\in I(K^n)$. 
Having this in mind we introduce, for any $1\le i \le n$, the monic polynomials $G_i=G_i(t_1,\dots t_i)$ of minimal degree in $R_i=K[t_1;\si_1,\de_1]\dots [t_i;\si_i,\de_i]$ such that $G_i\in I(K^n)$.  With these notations, we can now state the following proposition.
\begin{proposition}
	\label{I(K^n)}
	Let $K$ be a finite ring and, for $1\le i \le n$, consider the polynomials $G_i\in R_i=K[t_1,\si_1,\de_1]\dots [t_i,\si_i,\de_i]$ defined above.
	Then 
	$$I(K^n)=R_1G_1(t_1)+R_2G_2(t_1,t_2)+ \cdots + R_nG_n(t_1,\dots,t_n).$$ 
\end{proposition}
\begin{proof}
	It is clear that $R_1G_1(t_1)+R_2G_2(t_1,t_2)+ \cdots + R_nG_n(t_n)\subseteq I(K^n)$.   Now, if a polynomial $f\in \mulo$ belongs to $I(K^n)$, then,  for any $(a_1,\dots, a_n)\in K^n$, $f(a_1,a_2,\dots,a_n)=0$.  Let us write the consecutive remainders obtained during the evaluation process as follows.
	For any $(a_1,\dots,a_n)\in K^n$, we have 
	
	$f(t_1,\dots, t_n)=q_1(t_1,\dots, t_n)(t_n-a_n)+ f_1(t_1,\dots, t_{n-1},a_n).$
	
	$f_1(t_1,\dots, t_{n-1},a_n)=q_2(t_1,\dots, t_{n-1})(t_{n-1}-a_{n-1})+ f_2(t_1,\dots, t_{n-2},a_{n-1},a_n).$
	
	Continuing this process, we get
	
	$f_{n-2}(t_1,t_2,a_3,\dots,a_n)=q_{n-1}(t_1,t_2)(t_{2}-a_{2})+ f_{n-1}(t_1,a_2\dots,a_n)$ and
	
	$f_{n-1}(t_1,a_2\dots,a_n)=q_n(t_1)(t_{1}-a_{1})+ f_n(a_1,\dots,a_{n-1},a_n).$
	
	Since $f(a_1,a_2,\dots ,a_n)=f_n(a_1,\dots,a_{n-1},a_n)=0$, we conclude that for any for any $(a_1,\dots, a_n\in K^n$, $f_{n-1}(t_1,a_2\dots,a_n)\in R_1(t_1-a_1)$, i.e. that 
	$f_{n-1}(t_1,a_2\dots,a_n)\in R_1G_1(t_1)$.   Similarly, we have that $f_{n-2}(t_1,t_2,a_3,\dots,a_n)-f_{n-1}(t_1,a_2\dots,a_n)\in R_2\
	\cap I(K^n)=R_2G_2(t_1,t_2)$.  An easy induction yields that for any $1\le i \le n$, we have $f_{n-i}(t_1,\dots 
	t_i,a_{i+1},\dots, a_n)-f_{n-i+1}(t_1,\dots, 
	t_{i-1},a_{i},\dots, a_n)\in R_iG_i$.  From this we get successively that $f\in R_nG_n+f_1 \subseteq R_nG_n +R_{n-1}G_{n-1}+ f_2\subseteq R_nG_n +R_{n-1}G_{n-1}+R_{n-3}G_{n-3}+ f_3 \subseteq R_1G_1(t_1)+R_2G_2(t_1,t_2)+ \cdots + R_nG_n(t_1,\dots,t_n)$. This yields the desired result.
	%
\end{proof}
\section{ Skew  Reed-Muller Codes Using Iterated Skew Polynomial Rings }
In the first part of this section, we give some  preliminaries on Reed-Muller codes over the commutative polynomial ring.\\
Let $m$ be a positive integer and let $P_{1}, P_{2}, \ldots, P_{n}$ be the $n=q^{m}$ points in the affine space $\mathbb{A}^{m}\left(\mathbb{F}_{q}\right) .$ For any integer $r$ with $0 \leq r \leq m(q-1)$, let $\mathcal{R}_m\subseteq \mathbb{F}_{q}\left[x_{1}, x_{2}, \ldots, x_{m}\right]/(x_1^p-x_1,x_2^p-x_2,\cdots,x_m^p-x_m)$ be the set representing polynomials of degree less or equal than $r$. 
\begin{definition}
	Let $f=\sum_{i=0}^{n} a_{I} X^{I}$ be a polynomial in $\mathbb{F}_{q}\left[x_{1}, x_{2}, \ldots, x_{m}\right],$ where for $I=(i_1,\dots,i_n$) we define $X^I$ to be $x^{i_1}x^{i_2} \cdots x^{i_n}$. The operator evaluation $\operatorname{ev}(f)$ at $b=(b_1,b_2\ldots b_n) \in \mathbb{F}^n_{q}$ is defined as
	$$
	\begin{aligned}
	\mathcal{R}_m & \rightarrow \mathbb{F}_{q} \\
	\operatorname{ev}: \quad f=\sum_{i\in I} a_{i} X^{i} & \mapsto f(b)=\sum_{i\in I} a_{i} (b^i).
	\end{aligned}
	$$
\end{definition}
The Reed-Muller codes of the length $n=q^m$ can be obtained as evaluation of the polynomial $f$ in $R_m$  in all points $P_1, P_2,\cdots,P_n$ in the affine sapce $\mathbb{A}^{m}\left(\mathbb{F}_{q}\right) .$ 	The Reed-Muller codes of order $ r $  is defined as
$$
\mathcal{R}_{q}(r, m)=\left\{\left(f\left(P_{1}\right), f\left(P_{2}\right), \ldots, f\left(P_{n}\right)\right) \mid f \in R_m\right\}=\left\{\operatorname{ev}(f): f \in R_m , \operatorname{deg} f \leqslant r\right\}.
$$


The function 	$\operatorname{ev}:  \mathcal{R}_m \rightarrow \mathbb{F}_{q} $ is an homomorphism of algebras (see\cite{bardet}). Hence the Reed-Muller codes is generated by the codewords	$\operatorname{ev}(f) $ where $f$ is a monomial of degree less than or equal to $r$ in the set of monomials in $m$ variables.  It will be denoted 
$$
\mathcal{M} =\left\{ X_1^{a_1}X_2^{a_2}\dots X_n^{a_n}\right\}.
$$
\subsection{Iterated skew polynomial and Reed Muller codes}

In a recent paper, skew Reed-Muller codes were defined via the use of iterated Ore extensions with three variables over a finite field \cite{bib1}.  The authors used inner derivations and inner automorphisms, and (cf. Remarks \ref{First remarks}) these can be erased by a change of variables.  Their idea was to used Gr\"obner bases to make computations.  They used the available software on the subject.  But their methods used the ideal $I$ in our notations and hence they faced the problem that $I$ might be equal to the whole iterated Ore extension $R$.   Some of their computations were wrong and for the sake of future works, we will briefly correct them. In the lines below the variables $X_i$ are the one used in their paper.
%

\textbf{1.} Let $\mathbb F_4=\mathbb F_2(\alpha)$,  where 
$\alpha^2=\alpha + 1$.  In the ring $R_{1}=\mathbb{F}_{4}\left[Y_{1} ; 
\theta\right]$ where $\theta$ is the Frobenius automorphism, we put $X_1=Y_1+1$ and we have the commutation relation
$$
X_{1} \alpha=(Y_1+1)\alpha=\theta(\alpha)Y_1+\alpha=
\theta(\alpha) X_{1}+ \theta(\alpha) +\alpha=\alpha^{2} X_{1}+1.
$$
\noindent \textbf{2.} In the ring $R_{2}=R_{1}\left[Y_{2} \right]$ we put $X_2=\alpha^2Y_2+Y_1+\alpha^2$ so that we have $R_2=R_1[X_2;\theta_\alpha,\delta_{X_1+\alpha}]$, where $\theta_\alpha(x)=\alpha^{-1}x\alpha$ for any $x\in R_1$.  In $R_2$ we have the above commutation relation $X_{1} \alpha=\alpha^{2} X_{1}+1$ together with
$$
X_{2} \alpha =\alpha X_{2}+X_{1}+1 \quad {\rm and} \quad 
X_2X_1=\alpha X_1X_2 +\alpha^2 X_2+\alpha^2X_1^2 +\alpha^2.
$$

\noindent \textbf{3.} In the ring $R_3=R_2[Y_3]$, we put $X_3=\alpha^2Y_3+\alpha Y_1 +\alpha$.  We then have  
$R_{3}=R_{2}\left[X_{3} ; \theta_{\alpha}, \delta_{\alpha 
	X_{1}}\right]$ where, as above $\theta_\alpha(x)=\alpha^{-1}x\alpha$, for any $x\in \mathbb{F}_4$.  We have the above commutation relations together with
$$
\begin{aligned}
X_{3} \alpha &=\alpha X_{3}+\alpha X_{1}+\alpha, \\
X_{3} X_{1} &=\alpha X_{1} X_{3}+\alpha^2 X_{3}+\alpha^2X_1^2+\alpha^2X_1, \\
X_{3} X_{2} &=X_{2} X_{3}+\alpha^2 X_1X_3+\alpha^2X_3 
+X_1X_2+\alpha X_1^2 +\alpha^{2} X_{1}+X_2+1.
\end{aligned}
$$
Of course, the monomials in $X_1,X_2,X_3$ are not monomials in $Y_1,Y_2,Y_3$ and conversely.
The advantage of the variables $Y_1,Y_2,Y_3$ is clear.  We continue with these variables and consider,   

$$R_n=R_{n-1}[Y_{n}]=R_1[Y_2,Y_3,\ldots,Y_n]=\mathbb{F}_4[Y_1,\theta][Y_2,Y_3,\ldots,Y_n].$$
As in the commutative case, there are polynomials that annihilate any element of $\mathbb{F}_q^n$.   Thanks to Theorem \ref{finite fields} and \ref{I(K^n)}, we have that, in our situation $I(K^n)=R_1G_1+R_2G_2+\cdots +R_nG_n$ where $G_1=Y_1^3-Y_1$ and $G_i=Y_i^4-Y_i$ for any $2\le i \le n$ (notice that the form of $G_1$ is due to the theorem \ref{finite fields}). 
This iterated skew polynomial ring leads to a new family of skew Reed-Muller code generated by the evaluations of a base of monomials in $ m $ variables and of degree less than or equal to $ r $ given by 
$$
\mathcal{M} =\left\{ Y_1^{l_1}Y_2^{l_2}\dots Y_n^{l_n}\right\}.
$$
By iterated skew polynomial rings given above we have the evaluation of $f\in R_m$  at $b=(b_1,b_2\ldots b_n) \in \mathbb{F}^n_{q}$ is defined as
$$
\begin{aligned}
R_m & \rightarrow \mathbb{F}_{q} \\
\operatorname{ev}: \quad f (X_1,X_2,\ldots,X_n) & \mapsto f(b_1,\ldots, b_n)=\sum_{i=0}^{n} P(b_2,\ldots,b_n)\theta^{\frac{i(i+1)}{2}}(b_1).
\end{aligned}
$$
Where $f (X_1,X_2,\ldots,X_n)=g(Y_1)P(Y_2,\ldots, Y_n)\in R_1G_1+R_2G_2+\cdots +R_nG_n$.
By the  iterated skew polynomial ring and  the evaluation we have the new  new family   of the skew Reed-Muller codes  in the following proposition.
\begin{proposition}
	The  skew Reed-Muller codes of parameters $ r $ and $ m $  over $R_{m}=\mathbb{F}_q[X_{1} ; \theta_{1}, \delta_{1}]\ldots [X_m ; \theta_m, \delta_m]$ are generated as
	$$
	\mathcal{R}(r, m) \stackrel{\operatorname{def}}{=}\left\{ \sum_{i=0}^{n} P(b_2,\ldots,b_n)\theta^{\frac{i(i+1)}{2}}(b_1): P \in R_m \text {  such that } \operatorname{deg} P \leqslant\operatorname{deg} G_n \right\}.
	$$	
\end{proposition}
The choice of the rows used to form the matrix generating the Reed-Muller codes 	$\mathcal{R}(r, m)$  consists in selecting the monomials $Y_1^{l_1}Y_2^{l_2}\ldots Y_n^{l_n}$ of degree less than or equal to $r=(n-1)p+1$. This choice is based on the order of the monomials.
The group of permutations of Reed-Muller codes $\mathcal{R}(r, m)$ is the set of transformations affines $x \rightarrow Ax+b$ where $A\in  \mathbb{F}_q^{m\times m}$ is an invertible matrix and $b\in\mathbb{F}_q^m$.
\subsection{Skew  Reed Muller Codes Using Polynomial maps in $n$ variables.}
In this section we give the construction of skew Reed-Muller codes using polynomial maps in $n$ variables. We know that  the Reed-Muller codes are  defined as the codes generated by the evaluations of a base of monomials in $ m $ variables and of degree less than or equal to $ r $, the set of monomials in $ m $ variables will be denoted
$$
\mathcal{M} =\left\{ t_1^{l_1}t_2^{l_2}\dots t_n^{l_n}\right\}.
$$
Notice that in the case of iterated Ore extensions $\mulo$ it is always possible to order the indeterminates of a word in $t_1,\dots, t_n$ so that  the only words we have to consider are the one of the form $t_1^{l_1}t_2^{l_2}\dots t_n^{l_n}$.
We now give the evaluation of some monomials 
in $R=K[t_1,\sigma_1,\delta_1][x_2,\sigma_2,\delta_2],\cdots,[x_n,\sigma_n,\delta_n] $ at $ (a_1,a_2,\cdots,a_n)\in K^n.$
\begin{enumerate}
	
	\item	Consider the polynomial ring $R=K[t_1;\si_1,\de_1][t_2;\si_2,\de_2]$ and let us evaluate the polynomials $t_1t_2$ 
	at $(a_1,a_2)\in K^2$.   We have $t_1t_2=t_1(t_2-a_2)+t_1a_2=t_1(t_2-a_2)
	+\sigma_1(a_2)t_1+\de_1(a_2)$.  This leads to
	$(t_1t_2)(a_1,a_2)=\si_1(a_2)a_1+\de_1(a_2)$.
	\item	Consider the polynomial ring $R=K[t_1;\si_1,\de_1][t_2;\si_2,\de_2][t_3;\si_3,\de_3]$ and let us evaluate the polynomials $t_1t_2t_3$ at $(a_1,a_2,a_3)\in K^3$.   We have $t_1t_2t_3=t_1t_2(t_3-a_3)+t_1t_2a_3=t_1t_2(t_3-a_3)+t_1(\si_2(a_3)t_2+\de_2(a_3))=
	t_1t_2(t_3-a_3)+t_1\si_2(a_3)t_2+t_1\de_2(a_3)=t_1t_2(t_3-a_3)+\si_1(\si_2(a_3))t_1t_2+\si_1(\de_2(a_3))t_1+\de_1(\de_2(a_3))$.  This leads to\\
	$(t_1t_2t_3)(a_1,a_2,a_3)= \si_1(\si_2(a_3))a_1a_2+\si_1(\de_2(a_3))a_1+\de_1(\de_2(a_3))$.
	\item Consider the polynomial ring $R=K[t_1;\si_1,\de_1][t_2;\si_2,\de_2][t_3;\si_3,\de_3][t_4;\si_4,\de_4]$.  We have 
	
	\noindent $(t_1t_2t_3t_4)(a_1,a_2,a_3,a_4)=\si_1(\si_2(\si_3(a_4)a_3)a_2)a_1+\de_1(\si_2(\si_3(a_4)a_2)a_3)+\\\si_1(\de_2(\si_3(a_4)a_3))a_1+\de_1(\de_2(\si_3(a_4)a_3))+\si_1(\si_2(\de_3(a_4))a_2)a_1+\de_1(\si_2(\de_3(a_4))a_2)+\si_1(\de_2(\de_3(a_4)))a_1+\de_1(\de_2(\de_3(a_4)))).$
	%
\end{enumerate} 
Polynomial maps in $n$ variables and  the evaluation leads to the new family of skew Reed-Muller codes in the following example.
\begin{example}
	In this example, we give the construction of the Reed-Muller codes with polynomial maps in $2$ or $3$ variables.  In these examples we used the Frobenius automorphsim for $\si_1$ and the identity for $\sigma_2$ and $\sigma_3$.  The derivations are all zeroes. 
	\begin{enumerate}
		\item The evaluation of the set monomial $\{1,t_1,t_2,t_1t_2\}$ over $R=\mathbb{F}_4[t_1;\si_1,\de_1][t_2;\si_2,\de_2]$ gives the Reed-Muller code with parameters $[16,4,8].$
		\item The evaluation of the set monomial $\{1,t_1,t_2,t_2t_1\}$ over $R=\mathbb{F}_4[t_1;\si_1,\de_1][t_2;\si_2,\de_2]$ gives the Reed-Muller code with parmeters $[16,4,7].$
		\item The evaluation of the set monomial $\{1,t_1,t_2,t_1,t_3,t_1t_3,t_2t_3,t_1t_2t_3\}$ over $R=\mathbb{F}_4[t_1;\si_1,\de_1][t_2;\si_2,\de_2][t_3;\si_3,\de_3]$ gives the Reed-Muller code with parameters $[64,8].$
	\end{enumerate}	
\end{example}

\bibliographystyle{amsalpha}

\end{document}